\documentclass[11pt,letter]{article}
\usepackage{fullpage}
\usepackage[utf8x]{inputenc}
\usepackage{amsmath, amsthm}
\usepackage{wrapfig,floatflt,graphicx,amssymb,textcomp,array,amsmath}
\usepackage{enumerate}
\usepackage{multirow}
\usepackage{tabularx}
\usepackage{color}
\usepackage{todonotes}
\usepackage[titletoc,title]{appendix}
\usepackage{url}
\usepackage[hidelinks]{hyperref}
\usepackage{lineno}

\newcommand{\etal}{{et~al.~\!}}
\newcommand{\arrg}{\mathcal{A}}

\usepackage{color,xcolor}

\definecolor{mycolor}{rgb}{0, 0, 0}

\title{Euclidean Maximum Matchings in the Plane---Local to Global\thanks{This research is supported by NSERC. A preliminary version of the paper (with slightly weaker bounds) has appeared in WADS 2021.}
}

\author{Ahmad Biniaz\thanks{School of Computer Science, University of Windsor, abiniaz@uwindsor.ca}
	\and Anil Maheshwari\thanks{School of Computer Science, Carleton University, \{anil, michiel\}@scs.carleton.ca}
	\and  Michiel Smid\footnotemark[2]
}

\date{}
\newtheorem{lemma}{Lemma}

\newtheorem{theorem}{Theorem}
\newtheorem{observation}{Observation}
\newtheorem*{problem*}{Problem}

\newtheorem*{invariant*}{Invariant}

\begin{document}
	\maketitle
	\begin{abstract}
		
Let $M$ be a perfect matching on a set of points in the plane where every edge is a line segment between two points. We say that $M$ is {\em globally maximum} if it is a maximum-length matching on all points. We say that $M$ is $k$-{\em local maximum} if for any subset $M'=\{a_1b_1,\dots,a_kb_k\}$ of $k$ edges of $M$ it holds that $M'$ is a maximum-length matching on points $\{a_1,b_1,\dots,a_k,b_k\}$. We show that local maximum matchings are good approximations of global ones.

Let $\mu_k$ be the infimum ratio of the length of any $k$-local maximum matching to the length of any global maximum matching, over all finite point sets in the Euclidean plane. It is known that $\mu_k\geqslant \frac{k-1}{k}$ for any $k\geqslant 2$.
We show the following improved bounds for $k\in\{2,3\}$: $\sqrt{3/7}\leqslant\mu_2< 0.93 $ and $\sqrt{3}/2\leqslant\mu_3< 0.98$. 
We also show that every pairwise crossing matching is unique and it is globally maximum. 

Towards our proof of the lower bound for $\mu_2$ we show the following result which is of independent interest: If we increase the radii of pairwise intersecting disks by factor $2/\sqrt{3}$, then the resulting disks have a common intersection.

	\end{abstract}
	
\section{Introduction}
{\color{mycolor}A {\em matching} in a graph is a set of edges without a common vertex. A {\em prefect matching} is a matching that covers every  vertex of the graph.}
A maximum-weight matching in an edge-weighted graph is a matching in which the sum of edge weights is maximized. Maximum-weight matching is among well-studied structures in graph theory and combinatorial optimization. It has been studied from both combinatorial and computational points of view in both abstract and geometric settings, see for example \cite{Alon1995,Avis1983,Bereg2019,Dyer1984,Duan2014,Edmonds1965,Edmonds1965b,Gabow1990,Gabow2018,Galil1986,Kuhn1955,Kuhn1956,Rendl1988}. Over the years, it has found applications in several areas such as scheduling, facility location, and network switching. It has also been used as a key subroutine in other optimization algorithms, for example, network flow algorithms \cite{Edmonds1972,Lawler1976}, maximum cut in planar graphs \cite{Hadlock1975}, and switch scheduling algorithms \cite{McKeown1996} to name a few. In the geometric setting, where vertices are represented by points in a Euclidean space and edges are line segments, the maximum-weight matching is usually referred to as the {\em maximum-length matching}. 

Let $P$ be a set of $2n$ distinct points in the plane, and let $M$ be a perfect matching on $P$ where every edge of $M$ is a straight line segment---{\color{mycolor}the underline graph is the complete geometric  graph with vertex set $P$}. We say that $M$ is {\em globally maximum} if it is a maximum-length matching on $P$.
For an integer $k\leqslant n$ we say that $M$ is $k$-{\em local maximum} if for any subset $M'=\{a_1b_1,\dots,a_kb_k\}$ of $k$ edges of $M$ it holds that $M'$ is a maximum-length matching on points $\{a_1,b_1,\allowbreak\dots,\allowbreak a_k,b_k\}$; in other words $M'$ is a maximum-length matching on the endpoints of its edges.
Local maximum matchings appear in local search heuristics for approximating global maximum matchings, see e.g. \cite{Arkin1998}.

It is obvious that any global maximum matching is locally maximum. On the other hand, local maximum matchings are known to be good approximations of global ones. 
Let $\mu_k$ be the infimum ratio of the length of any $k$-local maximum matching to the length of any global maximum matching, over all finite point sets in the Euclidean plane. For $k=1$, the ratio $\mu_1$ could be arbitrary small, because any matching is $1$-local maximum. For $k\geqslant 2$, however, it is known that $\mu_k\geqslant \frac{k-1}{k}$ (see e.g. \cite[Corollary 8]{Arkin1998}); this bound is independent of the Euclidean metric and it is valid for any edge-weighted complete graph. A similar bound is known for matroid intersection \cite[Corollary 3.1]{Lee2010}. We present improved bounds for $\mu_2$ and $\mu_3$; this is going to be the main topic of this paper.

\subsection{Our contributions}
The general lower bound $\frac{k-1}{k}$ implies that $\mu_2\geqslant 1/2$ and $\mu_3\geqslant 2/3$. {\color{mycolor}We use the geometry of the Euclidean plane and improve these bounds to $\mu_2\geqslant \sqrt{3/7}\approx~ 0.654$ and $\mu_3\geqslant \sqrt{3}/2\approx 0.866$. 
	For upper bounds, we exhibit point sets with 2- and 3-local maximum
	matchings for which $\mu_2<0.93$ and $\mu_3<0.98$.}
In the discussion at the end of this paper we show that analogous ratios for local minimum matchings could be arbitrary large.

For an edge set $E$, we denote by $w(E)$ the total length of its edges.
{\color{mycolor}First we prove a (weaker) lower bound of $1/\sqrt{2}$ for $\mu_3$. To obtain this bound} we prove that for any 3-local maximum matching $M$ it holds that $w(M)\geqslant w(M^*)/\sqrt{2}$ where $M^*$ is a global maximum matching for the endpoints of edges in $M$. To do so, we consider the set $D$ of diametral disks of edges in $M$. A recent result of Bereg \etal \cite{Bereg2019} combined with Helly's theorem \cite{Helly1923,Radon1921} implies that the disks in $D$ have a common intersection. We take a point in this intersection and connect it to endpoints of all edges of $M$ to obtain a star $S$. Then we show that $w(M^*)\leqslant w(S) \leqslant \sqrt{2}\cdot w(M)$, which proves the weaker lower bound. {\color{mycolor}To achieve the lower bound $\sqrt{3}/2$ we follow a similar approach but employ a recent result of Barabanshchikova and
	Polyanskii~\cite{Barabanshchikova2024} instead of \cite{Bereg2019}.
	
	Our proof approach for showing the lower bound $\sqrt{3/7}$ for $\mu_2$ is similar to that of  $1/\sqrt{2}$ for $\mu_3$.} However, our proof consists of more technical ingredients. We show that for any 2-local maximum matching $M$ it holds that $w(M)\geqslant \sqrt{3/7}\cdot w(M^*)$ where $M^*$ is a global maximum matching for the endpoints of edges of $M$. Again we consider the set $D$ of diametral disks of edges of $M$. A difficulty arises here because now the disks in $D$ may not have a common intersection, although they pairwise intersect. To overcome this issue we enlarge the disks in $D$ to obtain a new set of disks that have a common intersection. Then we take a point in this intersection and construct our star $S$ as before, and we show that $w(M^*)\leqslant w(S) \leqslant \sqrt{7/3}\cdot w(M)$. To obtain this result we face two technical complications: (i) we need to show that the enlarged disks have a common intersection, and (ii) we need to bound the distance from the center of star $S$ to endpoints of $M$. To overcome the first issue we prove that if we increase the radii of pairwise intersecting disks by factor $2/\sqrt{3}$ then the resulting disks have a common intersection; the factor $2/\sqrt{3}$ is the smallest that achieves this property. This result has the same flavor as the problem of stabbing pairwise intersecting disks with four points \cite{Carmi2018,Danzer1986,Har-Peled2018,Stacho1981}. To overcome the second issue we prove a result in distance geometry.

In a related result, which is also of independent interest, we show that every pairwise crossing matching is unique and it is globally maximum. To show the maximality we transform our problem into an instance of the ``multicommodity flows in planar graphs'' that was studied by Okamura and Seymour~\cite{Okamura1981} in 1981.

{\color{mycolor}The paper is organized as follow. In Section~\ref{k-section} we review the general lower bound on $\mu_k$. In Section~\ref{upper-bounds} we present upper bounds for $\mu_2$ and $\mu_3$. Sections~\ref{3-local-section} and~\ref{2-local-section} present better lower bounds for $\mu_3$ and $\mu_2$, respectively. Our result on pairwise-crossing matchings is given in Section~\ref{pairwise-crossing-section}.}

\subsection{Some related works}
From the computational point of view, Edmonds \cite{Edmonds1965,Edmonds1965b} gave a polynomial-time algorithm for computing weighted matchings in general graphs (the term {\em weighted matching} refers to both minimum-weight matching and maximum-weight matching). Edmonds' algorithm is a generalization of the Hungarian algorithm for weighted matching in bipartite graphs \cite{Kuhn1955,Kuhn1956}. There are several implementations of Edmonds' algorithm (see e.g. \cite{Gabow1990,Gabow1989,Galil1986,Lawler1976}) with the best known running time $O(mn+n^2\log n)$ \cite{Gabow1990,Gabow2018} where $n$ and $m$ are the number of vertices and edges of the graph. 
One might expect faster algorithms for the ``maximum-length matching'' in the geometric setting where vertices are points in the plane and any two points are connected by a straight line segment; we are not aware of any such algorithm. For general graphs, there is a linear-time $(1-\varepsilon)$-approximation of maximum-weight matching \cite{Duan2014}.

The analysis of maximum-length matching ratios has received attention in the past. In a survey by Avis \cite{Avis1983} it is shown that the matching obtained by a greedy algorithm (that picks the largest available edge) is a $1/2$-approximation of the global maximum matching (even in arbitrary weighted graphs). 
Alon, Rajagopalan, Suri \cite{Alon1995} studied non-crossing matchings, where edges are not allowed to cross each other. They showed that the ratio of the length of a maximum-length non-crossing matching to the length of a maximum-length matching is at least $2/\pi$; this ratio is the best possible. Similar ratios have been studied for non-crossing spanning trees, Hamiltonian paths and cycles \cite{Alon1995,Biniaz2019,Dumitrescu2010}. Bereg \etal \cite{Bereg2019} showed the following combinatorial property of maximum-length matchings: the diametral disks, introduced by edges of a maximum-length matching, have a common intersection. A somewhat similar property was proved by Huemer \etal \cite{Huemer2019} for bi-colored points.

\section{A lower bound for $k$-local maximum matchings}
\label{k-section}
For the sake of completeness, and to facilitate comparisons with our improved bounds, we repeat a proof of the general lower bound $\frac{k-1}{k}$, borrowed from \cite{Arkin1998}.

\begin{theorem}
	\label{k-local-thr}
	Every $k$-local maximum matching is a $\frac{k-1}{k}$-approximation of a global maximum matching for any $k\geqslant 2$.
\end{theorem}

\begin{proof}
	Consider any $k$-local maximum matching $M$ and a corresponding global maximum matching $M^*$. The union of $M$ and $M^*$ consists of even cycles and/or single edges which belong to both matchings. It suffices to show, for each cycle~$C$, that the length of edges in $C\cap M$ is at least $\frac{k-1}{k}$ times that of edges in $C\cap M^*$.  
	
	Let $e_0,e_1,\dots,e_{|C|-1}$ be the edges of $C$ that appear in this order. Observe that $|C|\geqslant 4$, and that the edges of $C$ alternate between $M$ and $M^*$. Let $C_M$ and $C_{M^*}$ denote the sets of edges of $C$ that belong to $M$ and $M^*$, respectively. If $|C|\leqslant 2k$ then $w(C_M)=w(C_{M^*})$ because $M$ is $k$-local maximum, and thus we are done. Assume that $|C|\geqslant 2k+2$. After a suitable shifting of indices we may assume that $C_M=\{e_i: i \text{ is even}\}$ and $C_{M^*}\allowbreak=\allowbreak\{e_i: i \text{ is odd}\}$. Since $M$ is $k$-local maximum, for each even index $i$ we have 
	\begin{linenomath*}
		\begin{equation}
			\label{eq1}
			\notag
			w(e_{i})+w(e_{i+2})+\dots+w(e_{i+2k-2})\geqslant w(e_{i+1})+w(e_{i+3})+\dots+w(e_{i+2k-3})
		\end{equation}
	\end{linenomath*}
	where all indices are taken modulo $|C|$. By summing this inequality over all even indices, every edge of $C_M$ appears exactly $k$ times and every edge of $C_{M^*}$ appears exactly $k-1$ times, and thus we get $k\cdot w(C_M)\geqslant (k-1)\cdot w(C_{M^*})$. 
\end{proof}

It is implied from Theorem~\ref{k-local-thr} that $\mu_2\geqslant 1/2$ and $\mu_3\geqslant 2/3$. To establish stronger lower bounds, we need to incorporate more powerful ingredients. We use geometry of the Euclidean plane and improve both lower bounds. 

{\color{mycolor}
	\section{Upper bounds for 2- and 3-local maximum matchings}
	\label{upper-bounds}
	
	It is somewhat challenging to find point sets for which 2-local and 3-local maximum matchings are not globally maximum. Consider the point set $P$ with six points $\{A,B,C,D,E,F\}$ in Figure~\ref{upper-bound-fig}(a). The edge set $M_1=\{AB, CD, EF\}$ is a 2-local maximum matching for $P$. The length of $M_1$ is less than $17.76$. The edge set $M_2=\{AF,BC,DE\}$ is another matching for $P$, and its length is larger than $19.1$. This implies that $\mu_2< 17.76/19.1 < 0.93$. In other words a 2-local maximum matching may not approximate a global one better than ratio $0.93$. 
	
	Now consider the point set $Q$ with eight points $\{A,B,C,D,E,F, G,H\}$ in Figure~\ref{upper-bound-fig}(b). The edge set $M_1=\{AB, CD, EF, GH\}$ is a 3-local maximum matching for $Q$ and its length is less than $13.235$. The edge set $M_2=\{AH,BC,\allowbreak DE,FG\}$ is another matching for $Q$, and its length is larger than $13.515$. This implies that $\mu_2< 13.235/13.515 < 0.98$.

	\begin{figure}[htb]
		\centering
		\setlength{\tabcolsep}{0in}
		$\begin{tabular}{cc}
			\multicolumn{1}{m{.4\columnwidth}}{\centering\includegraphics[width=.4\columnwidth]{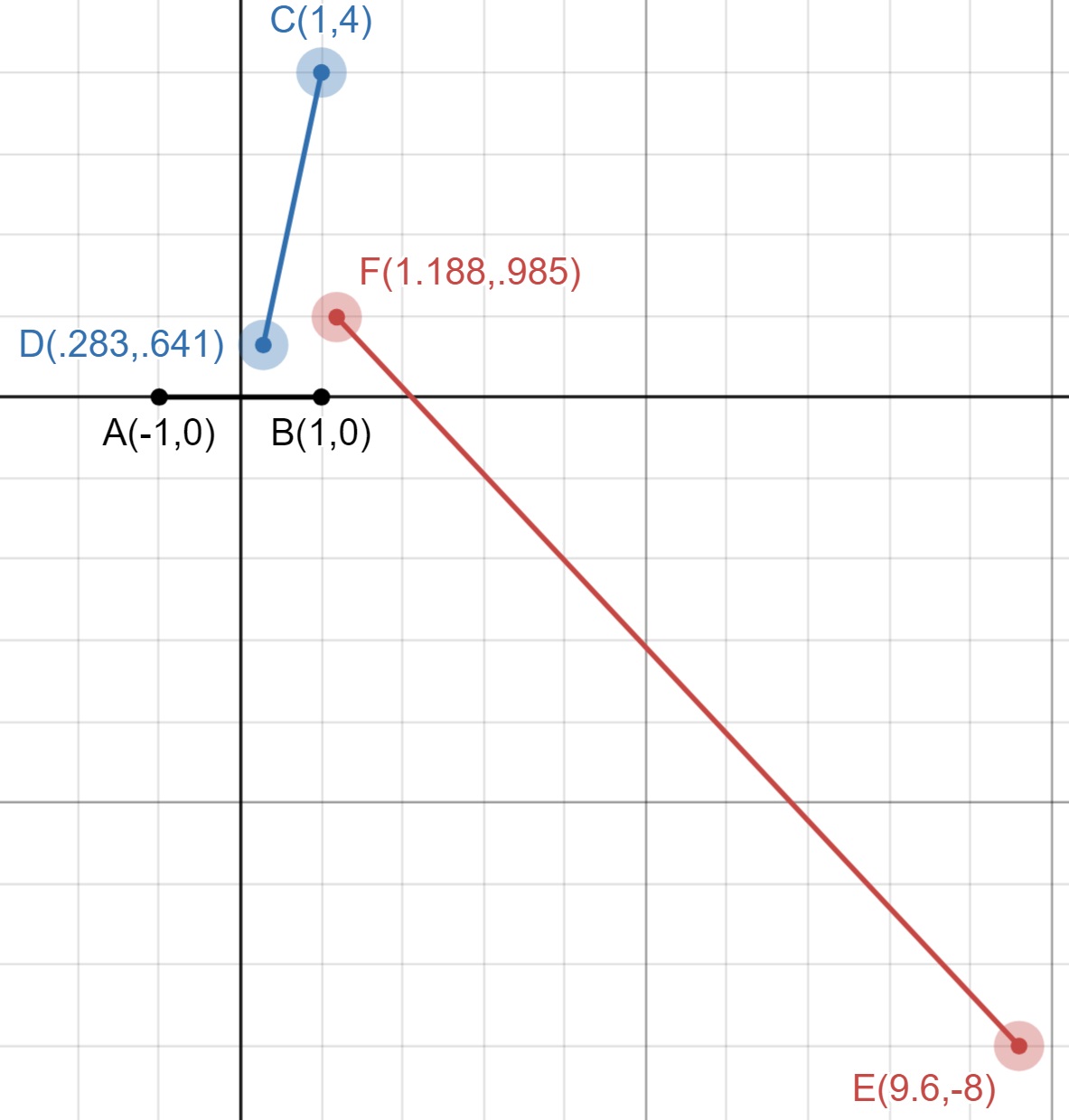}}
			&\multicolumn{1}{m{.6\columnwidth}}{\centering\includegraphics[width=.57\columnwidth]{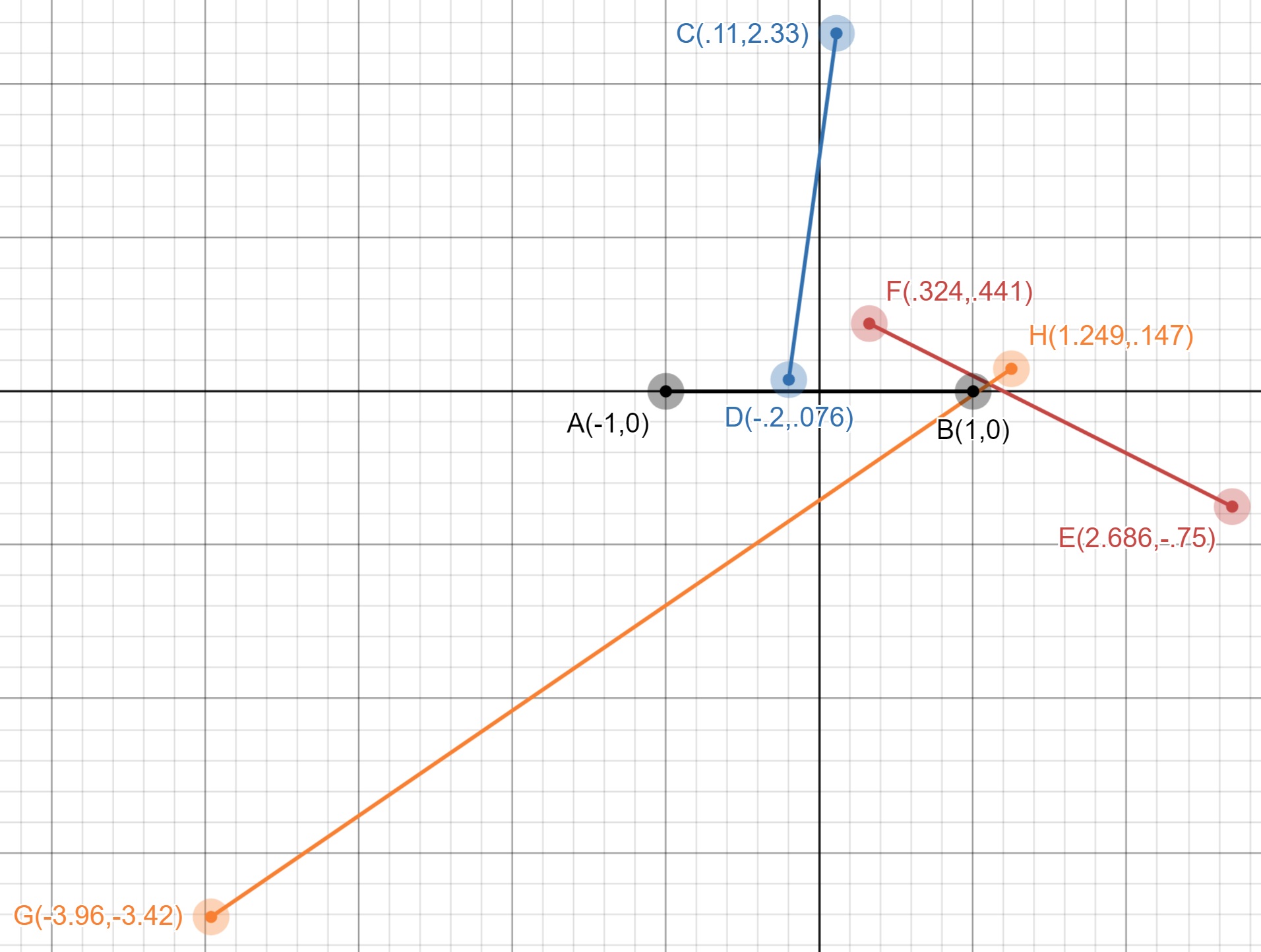}}\\
			(a)&(b)
		\end{tabular}$
		\caption{Illustration of upper bounds for (a) 2-local matchings and (b) 3-local matchings.}
		\label{upper-bound-fig}
	\end{figure}
}

\section{Better lower bound for 3-local maximum matchings}
\label{3-local-section}
{\color{mycolor}We give a proof of the lower bound $1/\sqrt{2}$ for 3-local maximum matchings first because it is easier to understand. Also it serves as a preliminary for our proof of the lower bound on $\mu_2$ which is given in Section~\ref{2-local-section}.}
Our proof benefits from the following result of Bereg \etal \cite{Bereg2019} and Helly's theorem \cite{Helly1923,Radon1921}.

\begin{theorem}[Bereg \etal \cite{Bereg2019}]
	\label{Bereg-thr}
	Consider any maximum matching of any set of six points in the plane. The diametral disks of the three edges in this matching have a nonempty intersection. 
\end{theorem}

\begin{theorem}[Helly's theorem in $\mathcal{R}^2$]
	\label{Helly-thr}
	If in a family of convex sets in the plane every triple of sets has a nonempty intersection, then the entire family has a nonempty intersection. 
\end{theorem}

\begin{figure}[ht]
	\centering
	\includegraphics[width=1.85in]{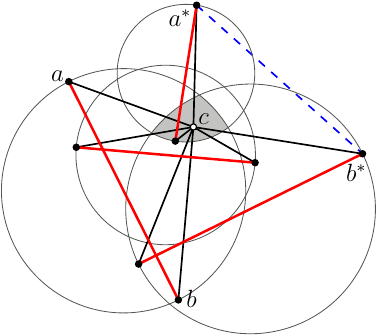} 
	\vspace{-10pt}
	\caption{Red edges belong to $M$, black edges belong to $S$, and  blue edge belongs to $M^*$.}
	\label{local-fig}
\end{figure}

\begin{theorem}
	\label{local3-thr}
	Every 3-local Euclidean maximum matching is a $\frac{1}{\sqrt{2}}$-approximation of a global Euclidean maximum matching.
\end{theorem}

\begin{proof}
	Consider any 3-local maximum matching $M$. Let $M^*$ be a global maximum matching for the endpoints of edges of $M$. Consider the set $D$ of diametral disks introduced by edges of $M$. Since $M$ is 3-local maximum, any three disks in $D$ have a common intersection (by Theorem~\ref{Bereg-thr}). With this property, it is implied by Theorem~\ref{Helly-thr} that the disks in $D$ have a common intersection (the shaded region in Figure~\ref{local-fig}). Let $c$ be a point in this intersection. Let $S$ be the star obtained by connecting $c$ to all endpoints of edges of $M$ as in Figure~\ref{local-fig}. Since $c$ is in the diametral disk of every edge $ab\in M$, it is at distance at most $|ab|/2$ from the midpoint of $ab$. By applying Lemma~\ref{endpoint-lemma} (which will be proved in Section~\ref{2-local-section}), with $c$ playing the role of $p$ and $r=1$, we have 
	\begin{linenomath*}	
		\begin{equation}
			\label{eq2}
			|ca|+|cb|\leqslant \sqrt{2} \cdot |ab|.
		\end{equation} 	
	\end{linenomath*}
	In Inequality~\eqref{eq2}, for every edge $ab \in M$, a unique pair of edges in $S$ is charged to $ab$. Therefore, $w(S) \leqslant \sqrt{2} \cdot w(M)$.  
	Now consider any edge $a^*b^*\in M^*$. By the triangle inequality we have that
	\begin{linenomath*}	 
		\begin{equation}
			\label{eq3}
			|a^*b^*|\leqslant |ca^*|+|cb^*|.
		\end{equation}
	\end{linenomath*}	
	In Inequality~\eqref{eq3}, every edge of $M^*$ is charged to a unique pair of edges in $S$. Therefore, $w(M^*) \leqslant w(S)$. Combining the two resulting inequalities we have that $w(M)\geqslant w(M^*)/\sqrt{2}$.
\end{proof}

{\color{mycolor}
	\vspace{8pt}
	\noindent{\bf A better lower bound.}
	In 1995, Fingerhut \cite{Fingerhut} conjectured that for any maximum-length matching $\{(a_1,b_1), \allowbreak \dots, \allowbreak(a_n,b_n)\}$ on any set of $2n$ points in the plane there exists a point $c$ such that 
	\begin{linenomath*}
		\begin{equation}
			\label{eq4}
			|a_i c| + |b_i c| \leqslant \alpha \cdot |a_i b_i| 
		\end{equation}
	\end{linenomath*}
	for all $i \in \{1,\ldots,n\}$, where $\alpha = 2/\sqrt{3}$. This conjecture is recently proved by Barabanshchikova and
	Polyanskii~\cite{Barabanshchikova2024}. An alternative interpretation of this result is that the ellipses with foci at $a_i$ and $b_i$ 
	and with eccentricity $\sqrt{3}/2$ have a nonempty intersection. This result combined with an argument similar to our proof of Theorem~\ref{local3-thr}, imply the approximation ratio $\frac{\sqrt{3}}{2}$ for 3-local maximum matchings. The main result of this section is summarized in the following theorem, which implies that $\mu_3\geqslant \sqrt{3}/2$.
	
	\begin{theorem}
		\label{local3-thr-revisited}
		Every 3-local Euclidean maximum matching is a $\frac{\sqrt{3}}{2}$-approximation of a global Euclidean maximum matching.
	\end{theorem}
}

\section{Better lower bound for 2-local maximum matchings}
\label{2-local-section}
In this section we prove that $\mu_2\geqslant \sqrt{3/7}\approx 0.65$, that is, 2-local maximum matchings are $\sqrt{3/7}$ approximations of global ones. Our proof approach employs an argument similar to the proof of the lower bound $1/\sqrt{2}$ for 3-local maximum matchings. Here we are facing an obstacle because diametral disks that are introduced by edges of a 2-local maximum matching may not have a common intersection. To handle this issue, we require stronger tools. Our idea is to increase the radii of disks---while preserving their centers---to obtain a new set of disks that have a common intersection. Then we apply our argument on this new set of disks. This gives rise to somewhat lengthier analysis. Also, two technical complications arise because now we need to show that the new disks have a common intersection, and we need to bound the total distance from any point in new disks to the endpoints of the corresponding matching edges. The following lemmas play important roles in our proof.
\begin{figure}[htb]
	
	\centering
	\includegraphics[width=.56\columnwidth]{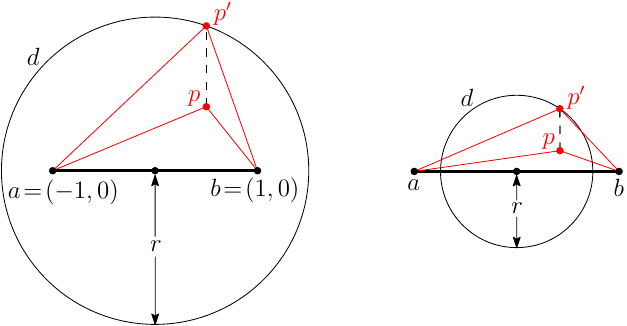}
	\caption{Illustration of the proof of Lemma~\ref{endpoint-lemma}.}
	\label{endpoint-fig}
\end{figure}
\begin{lemma} 
	\label{endpoint-lemma}
	Let $r>0$ be a real number. If $ab$ is a line segment in the plane and $p$ is a point at distance at most $\frac{r\cdot |ab|}{2}$ from the midpoint of $ab$ then 
	\begin{linenomath*}$$|pa|+|pb|\leqslant \sqrt{r^2+1}\cdot |ab|.$$
	\end{linenomath*}
\end{lemma}
\begin{proof}
	After scaling by factor $2/|ab|$ we will have $|ab|=2$ and $p$ at distance at most $r$ from the midpoint of $ab$. After a suitable rotation and translation assume that $a=(-1,0)$ and $b=(1,0)$. Any point $p=(x,y)$ at distance at most $r$ from the midpoint of $ab$ lies in the disk $d$ of radius $r$ that is centered at $(0,0)$ as in Figure~\ref{endpoint-fig}. Since $|ab|=2$, it suffices to prove that $|pa|+|pb|\leqslant 2\sqrt{r^2+1}$. Without loss of generality we may assume that $x\geqslant 0$ and $y\geqslant 0$. Let $p'$ be the vertical projection of $p$ onto the boundary of $d$ as in Figure~\ref{endpoint-fig}. Observe that $|pa|\leqslant |p'a|$ and $|pb|\leqslant |p'b|$. Thus the largest value of $|pa|+|pb|$ occurs when $p$ is on the boundary of $d$. Therefore, for the purpose of this lemma we assume that $p$ is on the boundary circle of $d$. 
	The circle has equation $x^2 + y^2 = r^2$. 
	Therefore, we can define $|pa|+|pb|$ as a function of $x$ as follows where $0\leqslant x\leqslant r$ (recall that $x$ is the $x$-coordinate of $p$, and $y$ is the $y$-coordinate of $p$).
	\begin{linenomath*}
		\begin{align}
			f(x)&=|pa|+|pb|\notag = \sqrt{(x+1)^2+y^2}+\sqrt{(x-1)^2+y^2} \notag \\ 
			&=\sqrt{x^2+y^2+1+2x}+\sqrt{x^2+y^2+1-2x}\notag \\ \notag
			&=\sqrt{r^2+1+2x}+\sqrt{r^2+1-2x}.\end{align}\end{linenomath*}
	We are interested in the largest value of $f(x)$ on interval  $x\in[0,r]$. By computing its derivative it turns out that $f(x)$ is decreasing on this interval. Thus the largest value of $f(x)$ is achieved at $x=0$, and it is $2\sqrt{r^2+1}$. 
\end{proof}

\begin{figure}[htb]
	\centering
	\includegraphics[width=1.0\columnwidth]{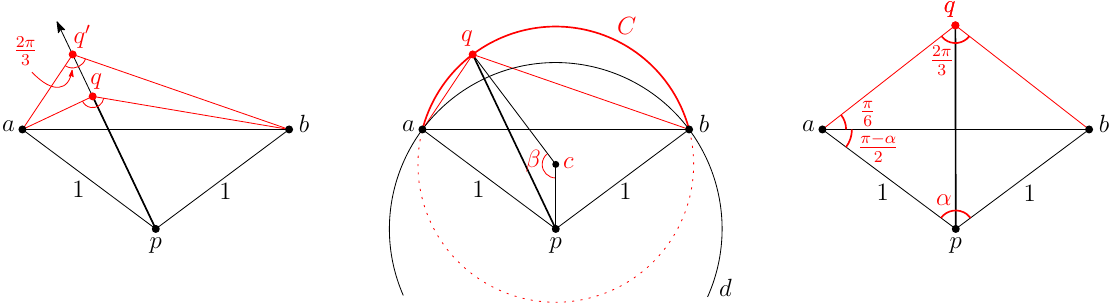}
	\caption{Illustration of the proof of Lemma~\ref{diameter-lemma}.}
	\label{diameter-fig}
\end{figure}

\begin{lemma} 
	\label{diameter-lemma}
	Let $a, p, b, q$ be the vertices of a convex quadrilateral that appear in this order along the boundary. If $|pa|=|pb|$ and $\angle aqb \geqslant 2\pi/3$ then $|pq|\leqslant \frac{2}{\sqrt{3}}|pa|$.
\end{lemma}

\begin{proof}
	After a suitable scaling, rotation, and reflection assume that $|pa|=1$, $ab$ is horizontal, and $p$ lies below $ab$ as in Figure~\ref{diameter-fig}-left. Since $|pa|=1$ in this new setting, it suffices to prove that $|pq|\leqslant 2/\sqrt{3}$. Consider the ray emanating from $p$ and passing through $q$. Let $q'$ be the point on this ray such that $\angle aq'b=2\pi/3$, and observe that $|pq'|\geqslant |pq|$. Thus for the purpose of this lemma we can assume that $\angle aqb = 2\pi/3$. 
	The locus of all points $q$, with $\angle aqb = 2\pi/3$, is a circular arc $C$ with endpoints $a$ and $b$. See Figure~\ref{diameter-fig}-middle. Let $c$ be the center of the circle that defines arc $C$. Since $ab$ is horizontal and $|pa|=|pb|$, the center $c$ lies on the vertical line through $p$.   
	Let $d$ be the disk of radius $1$ centered at $p$. If $c$ lies on or below $p$ then $C$ lies in $d$ and consequently $q$ is in $d$. In this case $|pq|\leqslant 1$, and we are done. Assume that $c$ lies above $p$ as in Figure~\ref{diameter-fig}-middle. By the law of cosines we have $|pq|=\sqrt{|pc|^2+|cq|^2-2|pc||cq|\cos\beta}$ where $\beta$ is the angle between segments $cp$ and $cq$. Since $|pc|$ and $|cq|$ are fixed for all points $q$ on $C$, the largest value of $|pq|$ is attained at $\beta=\pi$. Again for the purpose of this lemma we can assume that $\beta=\pi$, in which case $|qa|=|qb|$. Let $\alpha$ denote the angle between segments $pa$ and $pb$. Define $f(\alpha)=|pq|$ where $0\leqslant \alpha\leqslant \pi$. Recall that $\angle aqb=2\pi/3$. This setting is depicted in Figure~\ref{diameter-fig}-right. By the law of sines we have
	\begin{linenomath*}
		$$f(\alpha)=|pq|=\frac{\sin\left(\frac{\pi}{6}+\frac{\pi-\alpha}{2}\right)}{\sin\left(\frac{\pi}{3}\right)}
		=\frac{2\sin\left(\frac{4\pi-3\alpha}{6}\right)}{\sqrt{3}},$$\end{linenomath*}
	where $0\leqslant \alpha\leqslant \pi$. By computing the derivative of $f(\alpha)$ it turns out that its largest value is attained at $\alpha=\pi/3$, and it is $2/\sqrt{3}$.
\end{proof}

\begin{theorem}
	\label{stretch-lemma}
	Let $D$ be a set of pairwise intersecting disks. Let $D'$ be the set of disks obtained by increasing the radii of all disks in $D$ by factor $2/\sqrt{3}$ while preserving their centers. Then all disks in $D'$ have a common intersection. The factor $2/\sqrt{3}$ is tight.  
\end{theorem}

\begin{proof}
	It suffices to show that any three disks in $D'$ have a common intersection because afterwards Theorem~\ref{Helly-thr} implies that all disks in $D'$ have a common intersection. Consider any three disks $d'_1$, $d'_2$, $d'_3$ in $D'$ that are centered at $c_1$, $c_2$, $c_3$, and let $d_1$, $d_2$, $d_3$ be their corresponding disks in $D$. If $d_1$, $d_2$, $d_3$ have a common intersection, so do $d'_1$, $d'_2$, and $d'_3$. Assume that $d_1$, $d_2$, $d_3$ do not have a common intersection, as depicted in Figure~\ref{stretch-fig}. Let $u$ be the innermost intersection point of boundaries of $d_1$ and $d_2$, $v$ be the innermost intersection point of boundaries of $d_2$ and $d_3$, and $w$ be the innermost intersection point of boundaries of $d_3$ and $d_1$, as in Figure~\ref{stretch-fig}. We show that the Fermat point of triangle $\bigtriangleup uvw$ lies in all disks $d'_1$, $d'_2$, and $d'_3$. This would imply that these three disks have a common intersection. The Fermat point of a triangle is a point that minimizes the total distance to the three vertices of the triangle. If all angles of the triangle are less than $2\pi/3$ the Fermat point is inside the triangle and makes angle $2\pi/3$ with every two vertices of the triangle. If the triangle has a vertex of angle at least $2\pi/3$  the Fermat point is that vertex. 
	
	\begin{figure}[htb]
		\centering
		\includegraphics[width=.8\columnwidth]{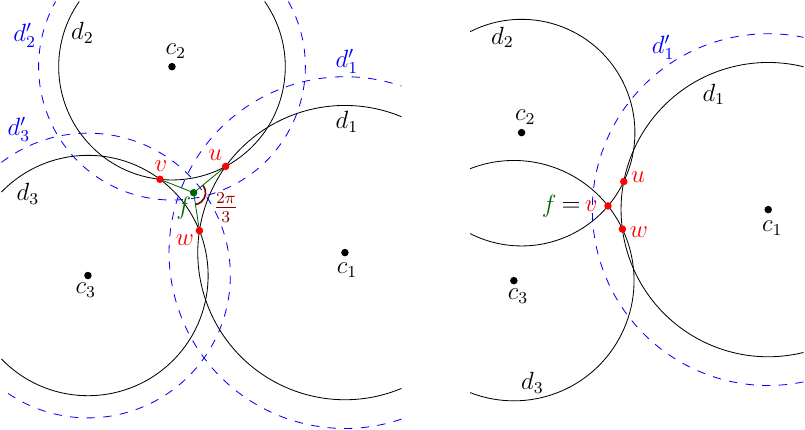}
		\caption{Illustration of the proof of Theorem~\ref{stretch-lemma}}
		\label{stretch-fig}
		\vspace{-12pt}
	\end{figure}
	
	Let $f$ be the Fermat point of $\bigtriangleup uvw$. First assume that all angles of $\bigtriangleup uvw$ are less than $2\pi/3$, as in Figure~\ref{stretch-fig}-left. In this case $f$ is inside $\bigtriangleup uvw$ and $\angle ufw=\angle wfv=\angle vfu=2\pi/3$. By Lemma~\ref{diameter-lemma} we have $|c_1f|\leqslant \frac{2}{\sqrt{3}}|c_1u|$ ($w, c_1,u,f$ play the roles of $a,p,b,q$ in the lemma, respectively). This and the fact that the radius of $d'_1$ is $\frac{2}{\sqrt{3}}|c_1u|$ imply that $f$ lies in $d'_1$. Analogously, we can show that $f$ lies in $d'_2$ and $d'_3$. This finishes our proof for this case.
	
	Now assume that one of the angles of $\bigtriangleup uvw$, say the angle $\angle uvw$ at $v$, is at least $2\pi/3$; see Figure~\ref{stretch-fig}-right. In this case $f=v$. Since $f$ is on the boundaries of $d_2$ and $d_3$, it lies in $d'_2$ and $d'_3$. By Lemma~\ref{diameter-lemma} we have $|c_1f|\leqslant \frac{2}{\sqrt{3}}|c_1u|$. Similarly to the previous case, this implies that $f$ lies in $d'_1$. This finishes our proof. 
	
	The factor $2/\sqrt{3}$ in the theorem is tight in the sense that if we replace it by any smaller constant then the disks in $D'$ may not have a common intersection. To verify this consider three disks of the same radius that pairwise touch (but do not properly intersect). For example assume that $d_1$, $d_2$, $d_3$ in Figure~\ref{stretch-fig}-left have radius $1$ and pairwise touch at $u$, $v$, and $w$. In this case $d'_1$, $d'_2$, $d'_3$ have radius $2/\sqrt{3}$. Moreover $\angle wc_1u=\allowbreak \angle uc_2v=\allowbreak\angle vc_3w=\pi/3$ and $f$ is inside $\bigtriangleup uvw$. In this setting $|c_1f|=|c_2f|=|c_3f|=2/\sqrt{3}$. This implies that $f$ is the only point in the common intersection of $d'_1$, $d'_2$ and $d'_3$. Therefore, if the radii of these disks are less than $2/\sqrt{3}$ then they wouldn't have a common intersection.
\end{proof}

\begin{theorem}
	\label{local2Euclidean-thr}
	Every 2-local Euclidean maximum matching is a $\sqrt{3/7}$ approximation of a global Euclidean maximum matching.
\end{theorem}
\begin{proof}
	Our proof approach is somewhat similar to that of Theorem~\ref{local3-thr}.
	Consider any 2-local maximum matching $M$. Let $M^*$ be a global maximum matching for the endpoints of edges of $M$. It is well known that that the two diametral disks introduced by the two edges of any maximum matching, on any set of four points in the plane, intersect each other (see e.g. \cite{Bereg2019}). Consider the set $D$ of diametral disks introduced by edges of $M$. Since $M$ is 2-local maximum, any two disks in $D$ intersect each other. However, all disks in $D$ may not have a common intersection. We increase the radii of all disks in $D$ by factor $2/\sqrt{3}$ while preserving their centers. Let $D'$ be the resulting set of disks. By Theorem~\ref{stretch-lemma} the disks in $D'$ have a common intersection. 
	Let $c$ be a point in this intersection. Let $S$ be the star obtained by connecting $c$ to all endpoints of edges of $M$. Consider any edge $ab\in M$, and let $d$ be its diametral disk in $D$ and $d'$ be the corresponding disk in $D'$. The radius of $d'$ is $\frac{2}{\sqrt{3}}\cdot \frac{|ab|}{2}$. Since $c$ is in $d'$, its distance from the center of $d'$ (which is the midpoint of $ab$) is at most $\frac{2}{\sqrt{3}}\cdot \frac{|ab|}{2}$. By applying Lemma~\ref{endpoint-lemma}, with $p=c$ and $r=2/\sqrt{3}$, we have
	$|ca|+|cb|\leqslant \sqrt{7/3} \cdot |ab|$.
	This implies that $w(S)\leqslant \sqrt{7/3}\cdot w(M)$. 
	For any edge $a^*b^*\in M^*$, by the triangle inequality we have $|a^*b^*|\leqslant |ca^*|+|cb^*|$, and thus $w(M^*)\leqslant w(S)$. Therefore, $w(M)\geqslant \sqrt{3/7}\cdot w(M^*)$.
\end{proof}

\section{Pairwise-crossing matchings are globally maximum}
\label{pairwise-crossing-section}
A pairwise crossing matching is a matching in which every pair of edges cross each other. It is easy to verify that any pairwise crossing matching is 2-local maximum. We claim that such matchings are in fact global maximum. We also claim that pairwise crossing matchings are unique. Both claims can be easily verified for points in convex position. In this section we prove these claims for points in general position, where no three points lie on a line.
\begin{observation}
	\label{unique-obs}
	Let $M$ be a pairwise crossing perfect matching on a point set $P$. Then for any edge $ab\in M$ it holds that the number of points of $P$ on each side of the line through $ab$ is $(|P|-2)/2$.
\end{observation}
\begin{theorem}
	\label{unique-crossing-thr}
	A pairwise crossing perfect matching on a point set is unique if it exists.
\end{theorem}
\begin{proof}
	Consider any even-size point set $P$ that has a pairwise crossing perfect matching. For the sake of contradiction assume that $P$ admits two different perfect matchings $M_1$ and $M_2$ each of which is pairwise crossing. 
	The union of $M_1$ and $M_2$ consists of connected components which are single edges (belong to both $M_1$ and $M_2$) and even cycles. Since $M_1\neq M_2$, $M_1\cup M_2$ contains some even cycles. Consider one such cycle, say $C$. Let $C_1$ and $C_2$ be the sets of edges of $C$ that belong to $M_1$ and $M_2$ respectively. Observe that each of $C_1$ and $C_2$ is a pairwise crossing perfect matching for vertices of $C$. 
	
	\let\qed\relax\end{proof}

\vspace{0pt}
\begin{wrapfigure}{r}{1.43in} 
	\centering
	\vspace{-20pt} 
	\includegraphics[width=1.4in]{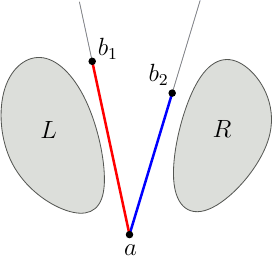} 
	\vspace{-30pt} 
\end{wrapfigure}

\vspace{-11pt}	
Let $a$ denote the lowest vertex of $C$; $a$ is a  vertex of the convex hull of $C$. Let $b_1$ and $b_2$ be the vertices of $C$ that are matched to $a$ via $C_1$ and $C_2$ respectively. After a suitable reflection assume that $b_2$ is to the right side of the line through $a$ and $b_1$ as in the figure to the right. Let $L$ be the set of vertices of $C$ that are to the left side of the line through $ab_1$, and let $R$ be the set of vertices of $C$ that are to the right side of the line through $ab_2$. Since $C_1$ is pairwise crossing, by Observation~\ref{unique-obs} we have $|L|=(|C|-2)/2$. Analogously we have $|R|=(|C|-2)/2$. Set $C'=L\cup R\cup\{a,b_1,b_2\}$, and observe that $C'\subseteq C$. Since the sets $L$, $R$, and $\{a,b_1,b_2\}$ are pairwise disjoint, $|C'|=|L|+|R|+3=|C|+1$. This is a contradiction because $C'$ is a subset of $C$. \hfill $\square$\vspace{8pt}

In Theorem~\ref{pairwise-crossing-thr} we prove that a pairwise crossing matching is globally maximum, i.e., it is a maximum-length matching for its endpoints. The following ``edge-disjoint paths problem'' that is studied by Okamura and Seymour~\cite{Okamura1981} will come in handy for our proof of Theorem~\ref{pairwise-crossing-thr}.
To state this problem in a simple way, we borrow some terminology from \cite{Wagner1995}.

Let $G=(V,E)$ be an embedded planar graph and let $N=\{(a_1,b_1), \dots,\allowbreak (a_k,b_k)\}$ be a set of pairs of distinct vertices of $V$ that lie on the outerface, as in Figure~\ref{crossing-fig}(a). A problem instance is a pair $(G,N)$ where the augmented graph $(V, E\cup\{a_1b_1,\dots,a_kb_k\})$ is Eulerian (i.e. it has a closed trail containing all edges). We note that the augmented graph may not be planar. The problem is to decide whether there are edge-disjoint paths $P_1,\dots,P_k$ in $G$ such that each $P_i$ connects $a_i$ to $b_i$.\footnote{This problem has applications in multicommodity flows in planar graphs \cite{Okamura1981}.} Okamura and Seymour~\cite{Okamura1981} gave a necessary and sufficient condition for the existence of such paths; this condition is stated below in Theorem~\ref{Okamura-Seymour-thr}. A {\em cut} $X$ is a nonempty proper subset of $V$. Let $c(X)$ be the number of edges in $G$ with one endpoint in $X$ and the other in $V\!\setminus\! X$, and let $d(X)$ be the number of pairs $(a_i,b_i)$ with one element in $X$ and the other in $V\!\setminus\! X$.
A cut $X$ is {\em essential} if the subgraphs of $G$ induced by $X$ and $V\!\setminus\! X$ are connected and neither set is disjoint with the outerface of $G$. If $X$ is essential then each of $X$ and $V\! \setminus\! X$ shares one single connected interval with the outerface; see Figure~\ref{crossing-fig}(a). 
\begin{theorem}[Okamura and Seymour, 1981] 
	\label{Okamura-Seymour-thr}
	An instance $(G,N)$ is solvable if and only if for any essential cut $X$ it holds that $c(X)-d(X)\geqslant 0$.
\end{theorem}
Wagner and Weihe \cite{Wagner1995} studied a computational version of the problem and presented a linear-time algorithm for finding edge-disjoint paths $P_1,\dots,P_k$. 

\begin{figure}[htb]
	\centering
	\setlength{\tabcolsep}{0in}
	$\begin{tabular}{cc}
		\multicolumn{1}{m{.45\columnwidth}}{\centering\includegraphics[width=.38\columnwidth]{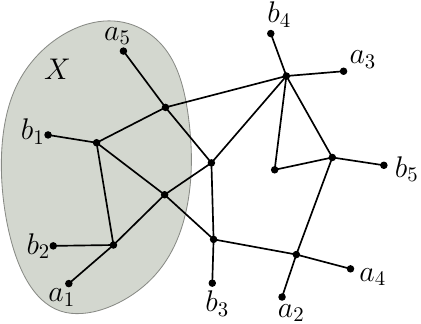}}
		&\multicolumn{1}{m{.55\columnwidth}}{\centering\includegraphics[width=.48\columnwidth]{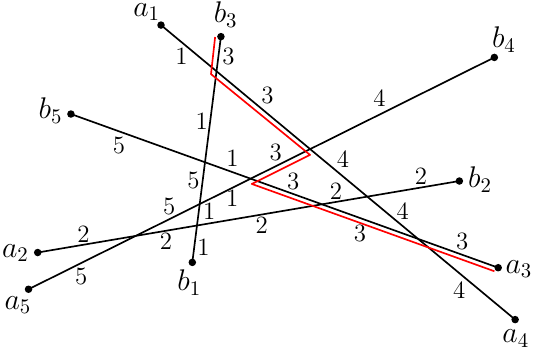}}\\
		(a)&(b)
	\end{tabular}$
	\caption{(a) An essential cut $X$ with $c(X)=4$ and $d(X)=2$. (b) Edge-disjoint paths between endpoints of edges of $M^*$.}
	\label{crossing-fig}
	
\end{figure}

\begin{theorem}
	\label{pairwise-crossing-thr}
	Any pairwise crossing matching is globally maximum.
\end{theorem}
\begin{proof}
	Consider any matching $M$ with pairwise crossing segments, and let $P$ be the set of endpoints of edges of $M$. Let $\arrg$ be the arrangement defined by the segments of $M$. Notice that $w(\arrg)=w(M)$, where $w(\arrg)$ is the total length of segments in $\arrg$. This arrangement is a planar graph where every vertex, that is a point of $P$, has degree 1 and every vertex, that is an intersection point of two segments of $M$, has degree 4 (assuming no three segments intersect at the same point). Now consider any perfect matching $M^*$ on $P$; $M^*$ could be a global maximum matching. Denote the edges of $M^*$ by $a_1b_1, a_2b_2,\dots$. To prove the theorem it suffices to show that $w(M^*)\leqslant w(\arrg)$. To show this inequality, we prove existence of edge-disjoint paths between all pairs $(a_i,b_i)$ in $\arrg$, as depicted in Figure~\ref{crossing-fig}(b). We may assume that $M$ and $M^*$ are edge disjoint because shared edges have the same contribution to each side of the inequality.

	Observe that the pair $(\arrg,M^*)$ is an instance of the problem of Okamura and Seymour~\cite{Okamura1981} because the augmented graph is Eulerian (here we slightly abuse $M^*$ to refer to a set of pairs). In the augmented graph, every point of $P$ has degree 2, whereas the degree of every other vertex is the same as its degree in $\arrg$. Consider any essential cut $X$ in $\arrg$. Set $X_P=X\cap P$. Consider the two sets $X_P$ and $P\!\setminus\!X_P$. Denote the smaller set by $Y_1$ and the larger set by $Y_2$. Notice that $|Y_1\cup Y_2|=|P|$, $|Y_1|\leqslant |P|/2$, and $|Y_2|\geqslant |P|/2$. We claim that no two points of $Y_1$ are matched to each other by an edge of $M$. To verify this claim we use contradiction. Assume that for two points $a$ and $b$ in $Y_1$ we have $ab\in M$. Since $X$ is essential, each of $Y_1$ and $Y_2$ consists of some points of $P$ that are consecutive on the outerface of $\arrg$. This and the fact that $M$ is pairwise crossing imply that all points of $Y_2$ lie on one side of the line through $ab$. This contradicts Observation~\ref{unique-obs}, and hence proves our claim.
	
	The above claim implies that every point in $Y_1$ is matched to a point in $Y_2$ by an edge of $M$. Any such edge of $M$ introduces at least one edge between $X$ and $\arrg\setminus\! X$ in $\arrg$. Therefore $c(X)\geqslant |Y_1|$.
	Since every $a_i$ and every $b_i$ belong to $P$, the number of pairs $(a_i,b_i)$ with one element in $X$ and another one in $\arrg\!\setminus\!X$ is the same as the number of such pairs with one element in $Y_1$ and the other in $Y_2$. The number of such pairs cannot be more than $|Y_1|$, and thus $d(X)\leqslant |Y_1|$. To this end we have that $c(X)\geqslant d(X)$. Having this constraint, Theorem~\ref{Okamura-Seymour-thr} implies that the instance $(\arrg,M^*)$ is solvable, and thus there are edge-disjoint paths between all pairs $(a_i,b_i)$. By the triangle inequality, $w(M^*)$ is at most the total length of these edge-disjoint paths, which is at most $w(\mathcal{A})$.
\end{proof}
\section{Discussion}
We believe that 3-local Euclidean maximum matchings are ``very good'' approximations of global Euclidean maximum matchings. {\color{mycolor}In particular we think that the lower bound on the length ratio should be closer to 0.98 than to $\sqrt{3}/2$.} 
A natural open problem is to use the geometry of the Euclidean plane and improve the lower bounds on the length ratios for 2- and 3-local maximum matchings. 

From the computational point of view, there are algorithms that compute a global maximum matching in polynomial time \cite{Gabow1990,Gabow2018,Gabow1989,Galil1986,Lawler1976} and there is a linear-time algorithm that gives a $ (1-\varepsilon) $-approximation \cite{Duan2014}. It would be interesting to see how fast a $k$-local maximum matching can be computed. Theorem~\ref{k-local-thr} suggests a local search strategy where repeatedly $k$-subsets of the current matching are tested for improvement. In its straightforward version this requires superlinear time. It would be interesting to see whether geometric insights could speed up the local search, maybe not (theoretically) matching the linear-time bound from \cite{Duan2014}, but leading to a practical and in particular simple algorithm. 

\begin{wrapfigure}{r}{1.3in} 
	\centering
	\vspace{-18pt} 
	\includegraphics[width=1.2in]{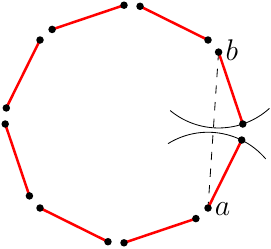} 
	\vspace{-15pt} 
\end{wrapfigure}
We note that analogous ratios for minimum-length matchings could be arbitrary large. In the figure to the right $2n$ points are placed on a circle such that distances between consecutive points are alternating between 1 and arbitrary small constant $\epsilon$. For a sufficiently large $n$, the red matching which has $n$ edges of length 1, would be 2-local minimum (the two arcs in the figure are centered at $a$ and $b$, and show that the length $|ab|$ is larger than the total length of two consecutive red edges). In this setting, the global minimum matching would have $n$ edges of length $\epsilon$. This shows that the ratio of the length of 2-local minimum matchings to that of global minimum matchings could be arbitrary large. By increasing the number of points (and hence flattening the perimeter of the circle) in this example, it can be shown that the length ratio of $k$-local minimum matchings could be arbitrary large, for any fixed $k\geqslant 2$.

{\color{mycolor}
	\section*{Acknowledgment} We thank Damien Robichaud, a former student at Carleton university, for finding and sharing with us the upper bound examples given in Section~\ref{upper-bounds}. These bounds are stronger than what we had initially.}
\bibliographystyle{abbrv}
\bibliography{MaxMatching-LocalGlobal}

\begin{thebibliography}{10}

\bibitem{Alon1995}
N.~Alon, S.~Rajagopalan, and S.~Suri.
\newblock Long non-crossing configurations in the plane.
\newblock {\em Fundamenta Informaticae}, 22(4):385--394, 1995.
\newblock Also in SoCG'93.

\bibitem{Arkin1998}
E.~M. Arkin and R.~Hassin.
\newblock On local search for weighted $k$-set packing.
\newblock {\em Mathematics of Operations Research}, 23(3):640--648, 1998.
\newblock Also in ESA'97.

\bibitem{Avis1983}
D.~Avis.
\newblock A survey of heuristics for the weighted matching problem.
\newblock {\em Networks}, 13(4):475--493, 1983.

\bibitem{Barabanshchikova2024}
P.~Barabanshchikova and A.~Polyanskii.
\newblock Intersecting ellipses induced by a max-sum matching.
\newblock {\em Journal of Global Optimization}, 88(2):395--407, 2024.

\bibitem{Bereg2019}
S.~Bereg, O.~Chac{\'{o}}n{-}Rivera, D.~Flores{-}Pe{\~{n}}aloza, C.~Huemer,
  P.~P{\'{e}}rez{-}Lantero, and C.~Seara.
\newblock On maximum-sum matchings of points.
\newblock {\em Journal of Global Optimization}, 85(1):111--128, 2023.

\bibitem{Biniaz2019}
A.~Biniaz, P.~Bose, K.~Crosbie, J.~D. Carufel, D.~Eppstein, A.~Maheshwari, and
  M.~H.~M. Smid.
\newblock Maximum plane trees in multipartite geometric graphs.
\newblock {\em Algorithmica}, 81(4):1512--1534, 2019.
\newblock Also in WADS'17.

\bibitem{Carmi2018}
P.~Carmi, M.~J. Katz, and P.~Morin.
\newblock Stabbing pairwise intersecting disks by four points.
\newblock {\em Discrete \& Computational Geometry}, 70(4):1751--1784, 2023.

\bibitem{Danzer1986}
L.~Danzer.
\newblock Zur {L}\"{o}sung des {G}allaischen {P}roblems \"{u}ber
  {K}reisscheiben in der {E}uklidischen {E}bene.
\newblock {\em Studia Scientiarum Mathematicarum Hungarica}, 21(1-2):111--134,
  1986.

\bibitem{Duan2014}
R.~Duan and S.~Pettie.
\newblock Linear-time approximation for maximum weight matching.
\newblock {\em Journal of the {ACM}}, 61(1):1:1--1:23, 2014.

\bibitem{Dumitrescu2010}
A.~Dumitrescu and C.~D. T{\'{o}}th.
\newblock Long non-crossing configurations in the plane.
\newblock {\em Discrete {\&} Computational Geometry}, 44(4):727--752, 2010.
\newblock Also in STACS'10.

\bibitem{Dyer1984}
M.~Dyer, A.~Frieze, and C.~McDiarmid.
\newblock Partitioning heuristics for two geometric maximization problems.
\newblock {\em Operations Research Letters}, 3(5):267--270, 1984.

\bibitem{Edmonds1965}
J.~Edmonds.
\newblock Maximum matching and a polyhedron with 0,1-vertices.
\newblock {\em Journal of Research of the National Bureau of Standards B},
  69:125--130, 1965.

\bibitem{Edmonds1965b}
J.~Edmonds.
\newblock Paths, trees, and flowers.
\newblock {\em Canadian Journal of Mathematics}, 17:449--467, 1965.

\bibitem{Edmonds1972}
J.~Edmonds and R.~M. Karp.
\newblock Theoretical improvements in algorithmic efficiency for network flow
  problems.
\newblock {\em Journal of the {ACM}}, 19(2):248--264, 1972.

\bibitem{Fingerhut}
D.~Eppstein.
\newblock Geometry junkyard.
\newblock \url{https://www.ics.uci.edu/\~{}eppstein/junkyard/maxmatch.html}.

\bibitem{Gabow1990}
H.~N. Gabow.
\newblock Data structures for weighted matching and nearest common ancestors
  with linking.
\newblock In {\em Proceedings of the First Annual {ACM-SIAM} Symposium on
  Discrete Algorithms {$($SODA$)$}}, pages 434--443, 1990.

\bibitem{Gabow2018}
H.~N. Gabow.
\newblock Data structures for weighted matching and extensions to
  \emph{b}-matching and \emph{f}-factors.
\newblock {\em {ACM} Transactions on Algorithms}, 14(3):39:1--39:80, 2018.

\bibitem{Gabow1989}
H.~N. Gabow, Z.~Galil, and T.~H. Spencer.
\newblock Efficient implementation of graph algorithms using contraction.
\newblock {\em Journal of the {ACM}}, 36(3):540--572, 1989.

\bibitem{Galil1986}
Z.~Galil, S.~Micali, and H.~N. Gabow.
\newblock An {O(EV} log {V)} algorithm for finding a maximal weighted matching
  in general graphs.
\newblock {\em {SIAM} Journal on Computing}, 15(1):120--130, 1986.
\newblock Also in FOCS'82.

\bibitem{Hadlock1975}
F.~Hadlock.
\newblock Finding a maximum cut of a planar graph in polynomial time.
\newblock {\em {SIAM} Journal on Computing}, 4(3):221--225, 1975.

\bibitem{Har-Peled2018}
S.~Har{-}Peled, H.~Kaplan, W.~Mulzer, L.~Roditty, P.~Seiferth, M.~Sharir, and
  M.~Willert.
\newblock Stabbing pairwise intersecting disks by five points.
\newblock In {\em 29th International Symposium on Algorithms and Computation,
  {ISAAC}}, pages 50:1--50:12, 2018.

\bibitem{Helly1923}
E.~Helly.
\newblock {\"{U}}ber {M}engen konvexer {K}\"{o}rper mit gemeinschaftlichen
  {P}unkten.
\newblock {\em Jahresbericht der Deutschen Mathematiker-Vereinigung},
  32:175--176, 1923.

\bibitem{Huemer2019}
C.~Huemer, P.~P{\'{e}}rez{-}Lantero, C.~Seara, and R.~I. Silveira.
\newblock Matching points with disks with a common intersection.
\newblock {\em Discrete Mathematics}, 342(7):1885--1893, 2019.

\bibitem{Kuhn1955}
H.~W. Kuhn.
\newblock The {H}ungarian method for the assignment problem.
\newblock {\em Naval Research Logistics Quarterly}, 2:83--97, 1955.

\bibitem{Kuhn1956}
H.~W. Kuhn.
\newblock Variants of the {H}ungarian method for assignment problems.
\newblock {\em Naval Research Logistics Quarterly}, 3:253--258, 1956.

\bibitem{Lawler1976}
E.~Lawler.
\newblock {\em Combinatorial optimization: networks and matroids}.
\newblock New York: Holt, Rinehart and Winston, 1976.

\bibitem{Lee2010}
J.~Lee, M.~Sviridenko, and J.~Vondr{\'{a}}k.
\newblock Submodular maximization over multiple matroids via generalized
  exchange properties.
\newblock {\em Mathematics of Operations Research}, 35(4):795--806, 2010.

\bibitem{McKeown1996}
N.~McKeown, V.~Anantharam, and J.~C. Walrand.
\newblock Achieving 100{\%} throughput in an input-queued switch.
\newblock In {\em Proceedings of the 15th {IEEE} {INFOCOM}}, pages 296--302,
  1996.

\bibitem{Okamura1981}
H.~Okamura and P.~D. Seymour.
\newblock Multicommodity flows in planar graphs.
\newblock {\em Journal of Combinatorial Theory, Series B}, 31(1):75--81, 1981.

\bibitem{Radon1921}
J.~Radon.
\newblock Mengen konvexer {K}\"{o}rper, die einen gemeinsamen {P}unkt
  enthalten.
\newblock {\em Mathematische Annalen}, 83(1):113--115, 1921.

\bibitem{Rendl1988}
F.~Rendl.
\newblock On the {E}uclidean assignment problem.
\newblock {\em Journal of Computational and Applied Mathematics}, 23(3):257 --
  265, 1988.

\bibitem{Stacho1981}
L.~Stach\'{o}.
\newblock A solution of {G}allai’s problem on pinning down circles.
\newblock {\em Matematikai Lapok}, 32(1-3):19--47, 1981/84.

\bibitem{Wagner1995}
D.~Wagner and K.~Weihe.
\newblock A linear-time algorithm for edge-disjoint paths in planar graphs.
\newblock {\em Combinatorica}, 15(1):135--150, 1995.
\newblock Also in ESA'93.

\end{thebibliography}
\end{document}